\documentclass[11pt,reqno]{amsart}

\usepackage{amssymb}

\usepackage{amsfonts}

\usepackage{amsmath}

\usepackage[all]{xy}

\newtheorem{theorem}{Theorem}[section]

\newtheorem{lemma}[theorem]{Lemma}

\newtheorem{proposition}[theorem]{Proposition}

\newtheorem{Definition}[theorem]{Definition}

\newtheorem{Example}[theorem]{Example}

\newtheorem{Remark}[theorem]{Remark}

\newenvironment{remark}{\begin{Remark}\begin{em}}{\end{em}\end{Remark}}

\newenvironment{definition}{\begin{Definition}\begin{em}}{\end{em}\end{Definition}}

\newcommand{\B}{\mathbf{B}}

\newcommand{\R}{{\mathbb R}}

\newcommand{\PP}{{\mathbb P}}

\newcommand{\uu}{\mathbf{u}}

\newcommand{\vv}{\mathbf{v}}

\newcommand{\ww}{\mathbf{w}}

\newcommand{\0}{\mathbf{0}}

\newcommand{\la}{\langle}

\newcommand{\ra}{\rangle}

\DeclareMathOperator{\gyr}{gyr}

\DeclareMathOperator{\tr}{tr}

\address{Sejong Kim, Department of Mathematics,
Chungbuk National University, Cheongju 361-763, Korea}
\email{skim@chungbuk.ac.kr}

\begin{document}

\author{Sejong Kim}

\title[Fidelity of M\"obius matrices related with Lorentz boosts]{Fidelity of M\"obius matrices related with \\ Lorentz boosts}

\date{}
\maketitle

\begin{abstract}
In this article we consider the extended version of a real counterpart of qubit density matrices, called the M\"obius matrix, and we see that it is a normalized Lorentz boost. Using the isomorphic gyrogroup structures between the set $\PP$ of all Lorentz boosts and the Einstein gyrogroup on the open unit ball $\B$ of $\R^n$ we give an explicit formula of the fidelity for M\"obius matrices in terms of Lorentz gamma factors.

\vspace{5mm}


\noindent {\bf PACS} (2010): 03.30.+p, 02.20.-a, 06.20.F-

\noindent {\bf Keywords}: M\"obius matrix, Lorentz boost, gyrogroup, Lorentz gamma factor, fidelity
\end{abstract}

\section{Introduction}

A qubit density matrix is a $2 \times 2$ positive semidefinite Hermitian matrix with trace $1$. It can be described by a Bloch vector
\begin{displaymath}
\vv = \left(
\begin{array}{ccc}
  v_{1} \\
  v_{2} \\
  v_{3} \\
\end{array}
\right), \ \Vert \vv \Vert \leq 1,
\end{displaymath}
where $\Vert \cdot \Vert$ is the Euclidean norm. In details,
\begin{displaymath}
\displaystyle \rho_{\vv} = \frac{1}{2} \left(
\begin{array}{cc}
  1 + v_3 & v_1 - i v_2 \\
  v_1 + i v_2 & 1 - v_3 \\
\end{array}
\right) = \frac{1}{2} (v_1 \sigma_x + v_2 \sigma_y + v_3 \sigma_z),
\end{displaymath}
where
\begin{displaymath}
\displaystyle \sigma_x = \left(
\begin{array}{cc}
  0 & 1 \\
  1 & 0 \\
\end{array}
\right), \sigma_y = \left(
\begin{array}{cc}
  0 & -i \\
  i & 0 \\
\end{array}
\right), \sigma_z = \left(
\begin{array}{cc}
  1 & 0 \\
  0 & -1 \\
\end{array}
\right)
\end{displaymath}
are Pauli matrices. It is known that all qubit pure states are parameterized by the unit sphere, while all qubit mixed states (or invertible density matrices) are parameterized by the open unit ball in $\R^3$. In the following we denote the open unit ball in $\R^{n}$ as $\B_{n}$ and consider column vectors $\vv$ in $\R^{n}$.

In general, it is difficult to extend the qubit mixed state $\rho_{\vv}$ to a density matrix that is parametrized by an $n$-dimensional Bloch vector $\vv = (v_{1}, v_{2}, \ldots, v_{n}) \in \B_n$ for $n > 3$. On the other hand, A. A. Ungar has suggested in \cite[Section 9.5]{Un08} the real counterpart of $\rho_{\vv}$ that shares similar properties with $\rho_{\vv}$ and its extended version such as
\begin{displaymath}
\begin{split}
\displaystyle \mu_{n, \vv} & = \frac{2 \gamma_{\vv}^2}{(n-3) + 4
\gamma_{\vv}^2} \left(
\begin{array}{cc}
  1 - \frac{1}{2 \gamma_{\vv}^2} & \vv^{T} \\
  \vv & \frac{1}{2 \gamma_{\vv}^2} I_{n} + \vv \vv^{T} \\
\end{array}
\right) \\
& = \frac{2 \gamma_{\vv}^2}{(n-3) + 4 \gamma_{\vv}^2} \left(
\begin{array}{cccccc}
  1 - \frac{1}{2 \gamma_{\vv}^2} & v_1 & v_2 & v_2 & \cdots & v_n \\
  v_1 & \frac{1}{2 \gamma_{\vv}^2} + v_1^2 & v_1 v_2 & v_1 v_3 & \cdots & v_1 v_n \\
  v_2 & v_1 v_2 & \frac{1}{2 \gamma_{\vv}^2} + v_2^2 & v_2 v_3 & \cdots & v_2 v_n \\
  v_3 & v_1 v_3 & v_2 v_3 & \frac{1}{2 \gamma_{\vv}^2} + v_3^2 & \cdots & v_3 v_n \\
  \vdots & \vdots & \vdots & \vdots & \ddots & \vdots \\
  v_n & v_1 v_n & v_2 v_n & v_3 v_n & \cdots & \frac{1}{2 \gamma_{\vv}^2} + v_n^2 \\
\end{array}
\right),
\end{split}
\end{displaymath}
where $I_{n}$ is the $n \times n$ identity matrix and $\displaystyle \gamma_{\vv} = \frac{1}{\sqrt{1 - \Vert \vv \Vert^2}}$ is known as the Lorentz gamma factor. This is called a \emph{M\"obius matrix} parameterized by the vector $\vv = (v_{1}, v_{2}, \ldots, v_{n}) \in \B_n$. Although it is not a natural extension of the qubit density
matrix, it is meaningful that we explore $\mu_{n, \vv}$ as a density matrix in the study of higher-level quantum states. The aim of this paper is to see the fidelity of M\"obius matrices, as one of the known measurements.

Lorentz was seeking the transformation under which Maxwell's equations were invariant when transformed from the ether to a moving frame. In 1905 Henri Poincar\'e recognized that the transformation has the properties of a mathematical group and named it after Lorentz. Later in the same year Albert Einstein derived the Lorentz transformation under the assumption of the principle of relativity and the constancy of the speed of light in any inertial reference frame. Lorentz transformation of the relativistically admissible vector is currently an important tool in special relativity, since it enables us to study relativistic mechanics in hyperbolic geometry. It also may include a rotation of space, and especially a rotation-free Lorentz transformation is called a \emph{Lorentz boost}. The Lorentz boost is a positive definite member of the Lorentz group $O(1,n)$, the group (under composition) of all linear transformations preserving the Lorentz form $\mathcal L$ defined by
\begin{displaymath}
\mathcal{L} \la(s, x_1, \ldots, x_n), (t, y_1, \ldots, y_n)\ra = - s
t + \sum_{i=1}^n x_i y_i.
\end{displaymath}
Indeed, the Lorentz boost is a member of the restricted Lorentz group $SO^+(1,n)$, the identity component of the Lorentz group consisting of all proper orthochronous maps.

In this paper, we can see the interesting result that the M\"obius matrix is a normalized Lorentz boost. So the study of M\"obius matrices will be associated with the algebraic structure of Lorentz boosts. In Section 2 we review a non-associative algebra structure (called a \emph{gyrogroup}) on the set of Lorentz boosts and provide an isomorphism with the Einstein gyrogroup on the open unit ball $\B_{n}$. In Section 3 we show that the M\"obius matrix is a normalized Lorentz boost via a diagonalization, and in Section 4 we calculate the fidelity of M\"obius matrices and give an explicit formula in terms of Lorentz gamma factors.

\section{Gyrogroup for Lorentz boosts}

We review first the Einstein's relativistic sum of admissible velocities of which magnitude is less than the speed of light $c \fallingdotseq 3 \times 10^5$ km/sec. In our purpose of this article, we assume the speed of light is normalized by the value 1, so that the admissible vectors are in the open unit ball
\begin{displaymath}
\B_{n} := \{ \vv \in \R^n: \Vert \vv \Vert < 1 \}.
\end{displaymath}
Then the relativistic sum of two admissible vectors $\uu$ and $\vv$ in $\B_{n}$ is given by
\begin{equation} \label{E:Einplus}
\uu \oplus \vv = \frac{1}{1 + \uu^{T} \vv} \left\{ \uu +
\frac{1}{\gamma_{\uu}} \vv + \frac{\gamma_{\uu}}{1 + \gamma_{\uu}}
(\uu^{T} \vv) \uu \right\},
\end{equation}
where $\gamma_{\uu}$ is the well-known \emph{Lorentz factor}
\begin{equation} \label{E:gamma}
\displaystyle \gamma_{\uu} = \frac{1}{\sqrt{1 - \Vert \uu \Vert^2}}.
\end{equation}
Note that $\uu^{T} \vv$ is just the Euclidean inner product of $\uu$ and $\vv$ written in matrix form.

\begin{definition}
The formula (\ref{E:Einplus}) defines a binary operation, called the \emph{Einstein velocity addition}, on the open unit ball $\B_{n}$ of $\R^n$.
\end{definition}

\begin{remark}
The Einstein addition $\uu \oplus \vv$ of two admissible vectors $\uu$ and $\vv$ in $\B_{n}$ may be alternatively obtained by applying the \emph{Lorentz boost}
\begin{equation} \label{E:Lboost}
B(\uu) = \left(
\begin{array}{cc}
  \gamma_{\uu} & \gamma_{\uu} \uu^{T} \\
  \gamma_{\uu} \uu & I + \frac{\gamma_{\uu}^{2}}{1 + \gamma_{\uu}} \uu \uu^{T} \\
\end{array}
\right)
\end{equation}
to $\left(
\begin{array}{c}
  \gamma_{\vv} \\
  \gamma_{\vv} \vv \\
\end{array}
\right)$ and obtaining
\begin{displaymath}
B(\uu) \left(
\begin{array}{c}
  \gamma_{\vv} \\
  \gamma_{\vv} \vv \\
\end{array}
\right) = \left(
\begin{array}{c}
  \gamma_{\uu \oplus \vv} \\
  \gamma_{\uu \oplus \vv} (\uu \oplus \vv) \\
\end{array}
\right),
\end{displaymath}
where we use the gamma identity $\gamma_{\uu \oplus \vv} = \gamma_{\uu} \gamma_{\vv} (1 + \uu^T \vv)$.
\end{remark}

To abstractly analyze Einstein velocity addition in the theory of special relativity, A. A. Ungar has introduced and studied in several papers and books structures that he has called gyrogroups; see \cite{Un08} and its bibliography. His algebraic axioms are reminiscent of those for a group, but a gyrogroup operation is neither associative nor commutative in general.

\begin{definition} \label{D:gyrogroup}
A triple $(G, \oplus, 0)$ is a \emph{gyrogroup} if the following axioms are satisfied for all $a, b, c \in G$.
\begin{itemize}
\item[(G1)] $0 \oplus a = a \oplus 0 = a$ (existence of identity);
\item[(G2)] $a \oplus (-a) = (-a) \oplus a = 0$ (existence of inverses);
\item[(G3)] There is an automorphism $\gyr[a,b] : G \to G$ for each $a, b\in G$ such that
 \begin{center}
 $a \oplus (b \oplus c) = (a \oplus b) \oplus \gyr[a,b]c$ (gyroassociativity);
 \end{center}
\item[(G4)] $\gyr[0,a] =$ id$_{G}$;
\item[(G5)] $\gyr[a \oplus b, b] = \gyr[a,b]$ (loop property).
\end{itemize}
A gyrogroup $(G, \oplus)$ is \emph{gyrocommutative} if it satisfies
 \begin{center}
 $a \oplus b = \gyr[a,b](b \oplus a)$ (gyrocommutativity).
 \end{center}
A gyrogroup is \emph{uniquely $2$-divisible} if for every $b \in G$, there exists a unique $a \in G$ such that $a \oplus a = b$.
\end{definition}

The map $\gyr[a,b]$ is called the \emph{gyroautomorphism} or \emph{Thomas gyration} generated by $a$ and $b$, which is analogous to the precession map in a loop theory. It has been shown in \cite{SSS} that gyrocommutative gyrogroups are equivalent to Bruck loops with respect to the same operation. It follows that uniquely $2$-divisible gyrocommutative gyrogroups are equivalent to $B$-loops, uniquely 2-divisible Bruck loops. J. Lawson and Y. Lim
have recently introduced dyadic symmetric sets in \cite{LL1} and showed the equivalence with uniquely $2$-divisible gyrocommutative gyrogroups. In our purpose of this article we follow the notion of gyrogroups.

A. A. Ungar has shown in \cite[Chapter 3]{Un08} by computer algebra that Einstein addition on the open unit ball $\B_{n}$ is a gyrocommutative gyrogroup operation, and the gyroautomorphisms are orthogonal transformations preserving the Euclidean inner product and the inherited norm. We call $(\B_{n}, \oplus)$ the \emph{Einstein
(gyrocommutative) gyrogroup}, where $\oplus$ is defined by the equation (\ref{E:Einplus}).

\begin{remark} \label{R:1/2v}
We note that the Einstein gyrogroup $(\B_{n}, \oplus)$ is uniquely $2$-divisible; for any $\vv \in \B_{n}$ there exists a unique
\begin{displaymath}
\displaystyle \ww = \frac{\gamma_{\vv}}{1 + \gamma_{\vv}} \vv \in \B
\end{displaymath}
such that $\ww \oplus \ww = \vv$ (see the equation (6.297) of \cite{Un08}). We denote it simply by $\ww := (1/2) \otimes \vv$, or $\vv = 2 \otimes \ww$.
\end{remark}

We now see the gyrogroup structure on the set $\PP$ of all Lorentz boosts given in the equation
\eqref{E:Lboost}. From the polar decomposition of $B(\uu) B(\vv)$ for $\uu, \vv \in
\B_{n}$ we have the relation
\begin{equation} \label{E:EinBoost}
B(\uu \oplus \vv) = \left( B(\uu) B(\vv)^2 B(\uu) \right)^{1/2},
\end{equation}
see \cite{KL2} for more details. Hence, we obtain
\begin{theorem} \label{T:iso1 boost}
The Lorentz boost map $B$ is an isomorphism from $(\B_{n}, \oplus, \0)$
to $(\PP, \star, I)$, where
\begin{displaymath}
B(\uu) \star B(\vv) = \left( B(\uu) B(\vv)^2 B(\uu) \right)^{1/2}.
\end{displaymath}
Furthermore, the powers and roots in $(\PP, \ast)$ agree with those
of matrix multiplication.
\end{theorem}

On the cone $\Omega$ of positive definite Hermitian matrices, the
squaring map $D: \Omega \to \Omega$, \ $D(A) = A^2$ gives us a
different algebraic structure on the set $\PP$. We note that the
squaring map $D$ is a bijection since any positive definite
Hermitian matrix has a unique square root in $\Omega$.
\begin{theorem} \label{T:iso2 boost}
The composition $D \circ B: (\B_{n}, \oplus, \0) \to (\PP, \ast, I)$ is
also an isomorphism, where
\begin{displaymath}
B(\uu) \ast B(\vv) = B(\uu)^{1/2} B(\vv) B(\uu)^{1/2}.
\end{displaymath}
\end{theorem}

\begin{remark} \label{R:Boost}
From Theorem \ref{T:iso1 boost} and Theorem \ref{T:iso2 boost} we
see that both $(\PP, \star, I)$ and $(\PP, \ast, I)$ are uniquely
$2$-divisible gyrocommutative gyrogroups. Moreover, we have
\begin{displaymath}
\displaystyle B(2 \otimes \vv) = B(\vv)^{2}, \ B((1/2) \otimes \vv) = B(\vv)^{1/2}
\end{displaymath}
for any $\vv \in \B_{n}$.
\end{remark}

\section{M\"obius matrices and Lorentz boosts}

First of all, we see the M\"obius matrix parameterized by the vector $\vv \in \B_{n}$, $n \geq 3$, as an extended version of the real counterpart of qubit density matrices:
\begin{equation} \label{E:Mobius}
\displaystyle \mu_{n, \vv} = \frac{2 \gamma_{\vv}^2}{(n-3) + 4
\gamma_{\vv}^2} \left(
\begin{array}{cc}
  1 - \frac{1}{2 \gamma_{\vv}^2} & \vv^{T} \\
  \vv & \frac{1}{2 \gamma_{\vv}^2} I_{n} + \vv \vv^{T} \\
\end{array}
\right).
\end{equation}
This is an $(n+1) \times (n+1)$ symmetric matrix, and we verify a diagonalization of M\"obius matrix.

\begin{theorem} \label{T:Diag Mobius}
For each $\vv \in \B_{n}$ there exist an orthogonal matrix $O_{\vv}$ and a diagonal matrix $D_{\vv}$
\begin{displaymath}
\begin{split}
\displaystyle O_{\vv} & = \left(
\begin{array}{ccccc}
  \frac{1}{\sqrt{2}} & -\frac{1}{\sqrt{2}} & 0 & \cdots & 0 \\
  \frac{1}{\sqrt{2} \Vert \vv \Vert} \vv & \frac{1}{\sqrt{2} \Vert \vv \Vert} \vv & \uu_{1} & \cdots & \uu_{n-1} \\
\end{array}
\right), \\
D_{\vv} & = \frac{1}{(n-3) + 4 \gamma_{\vv}^2} \left(
\begin{array}{cc}
  \begin{array}{cc}
    \lambda^2 & 0 \\
    0 & \frac{1}{\lambda^2} \\
  \end{array} & 0 \\
  0 & I_{n-1} \\
\end{array}
\right)
\end{split}
\end{displaymath}
such that $\mu_{n, \vv} = O_{\vv}^{T} D_{\vv} O_{\vv}$, where $\{ \uu_{j} \ : \ \vv^{T} \uu_{j} = 0$ for all $j = 1, 2, \dots, n-1 \}$ is an orthonormal set obtained by the Gram-Schmidt process, and
\begin{displaymath}
\displaystyle \lambda = \sqrt{\frac{1 + \Vert \vv \Vert}{1 - \Vert \vv \Vert}} > 1.
\end{displaymath}
\end{theorem}

\begin{proof}
Let
\begin{displaymath}
A = \frac{(n-3) + 4 \gamma_{\vv}^2}{2 \gamma_{\vv}^2} \mu_{n, \vv} =
\left(
\begin{array}{cc}
  1 - \frac{1}{2 \gamma_{\vv}^2} & \vv^{T} \\
  \vv & \frac{1}{2 \gamma_{\vv}^2} I_{n} + \vv \vv^{T} \\
\end{array}
\right)
\end{displaymath}
It is enough to show that
\begin{displaymath}
\displaystyle A = O_{\vv}^{T} \cdot \frac{1}{2 \gamma_{\vv}^2}
\left(
\begin{array}{cc}
  \begin{array}{cc}
    \lambda^2 & 0 \\
    0 & \frac{1}{\lambda^2} \\
  \end{array} & 0 \\
  0 & I_{n-1} \\
\end{array}
\right) \cdot O_{\vv}.
\end{displaymath}

Indeed,
\begin{displaymath}
\begin{split}
A \left(
\begin{array}{c}
  \frac{1}{\sqrt{2}} \\
  \frac{1}{\sqrt{2} \Vert \vv \Vert} \vv \\
\end{array}
\right) & = \left(
\begin{array}{c}
  \frac{1}{\sqrt{2}} \left( 1 - \frac{1}{2 \gamma_{\vv}^2} \right) + \frac{1}{\sqrt{2} \Vert \vv \Vert} \vv^{T} \vv \\
  \frac{1}{\sqrt{2}} \vv + \left( \frac{1}{2 \gamma_{\vv}^2} I_{n} + \vv \vv^{T} \right) \frac{1}{\sqrt{2} \Vert \vv \Vert} \vv \\
\end{array}
\right) \\
& = \left(
\begin{array}{c}
  \frac{1}{2 \sqrt{2}} (1 + \Vert \vv \Vert)^2 \\
  \frac{1}{2 \sqrt{2} \Vert \vv \Vert} (1 + \Vert \vv \Vert)^2 \vv \\
\end{array}
\right) \\
& = \frac{(1 + \Vert \vv \Vert)^2}{2} \left(
\begin{array}{c}
  \frac{1}{\sqrt{2}} \\
  \frac{1}{\sqrt{2} \Vert \vv \Vert} \vv \\
\end{array}
\right).
\end{split}
\end{displaymath}
Here, $\displaystyle \frac{(1 + \Vert \vv \Vert)^2}{2} = \frac{1}{2 \gamma_{\vv}^2} \cdot \frac{1 + \Vert \vv \Vert}{1 - \Vert \vv \Vert} = \frac{\lambda^2}{2 \gamma_{\vv}^2}$. Similarly,
\begin{displaymath}
A \left(
\begin{array}{c}
  -\frac{1}{\sqrt{2}} \\
  \frac{1}{\sqrt{2} \Vert \vv \Vert} \vv \\
\end{array}
\right) = \frac{(1 - \Vert \vv \Vert)^2}{2} \left(
\begin{array}{c}
  -\frac{1}{\sqrt{2}} \\
  \frac{1}{\sqrt{2} \Vert \vv \Vert} \vv \\
\end{array}
\right) = \frac{1}{2 \gamma_{\vv}^2 \lambda^2} \left(
\begin{array}{c}
  -\frac{1}{\sqrt{2}} \\
  \frac{1}{\sqrt{2} \Vert \vv \Vert} \vv \\
\end{array}
\right).
\end{displaymath}
Finally for each $j = 1, 2, \dots, n-1$
\begin{displaymath}
A \left(
\begin{array}{c}
  0 \\
  \uu_{j} \\
\end{array}
\right) = \left(
\begin{array}{c}
  \vv^{T} \uu_{j} \\
  \left( \frac{1}{2 \gamma_{\vv}^2} I_{n} + \vv \vv^{T} \right) \uu_{j} \\
\end{array}
\right) = \frac{1}{2 \gamma_{\vv}^2} \left(
\begin{array}{c}
  0 \\
  \uu_{j} \\
\end{array}
\right)
\end{displaymath}
since $\vv^{T} \uu_{j} = 0$.
\end{proof}

\begin{remark} \label{R:boost trace}
By Theorem \ref{T:Diag Mobius} we have that the matrix $\mu_{n, \vv}$ is positive definite,
\begin{displaymath}
\displaystyle \tr \mu_{n, \vv} = \frac{1}{(n-3) + 4 \gamma_{\vv}^2}
\left( \lambda^2 + \frac{1}{\lambda^2} + n-1 \right) = 1,
\end{displaymath}
and
\begin{displaymath}
\displaystyle \det \mu_{n, \vv} = \left( \frac{1}{(n-3) + 4 \gamma_{\vv}^2} \right)^{n+1} = \left( \frac{1 - \Vert \vv \Vert^2}{(n+1) - (n-3) \Vert \vv \Vert^2} \right)^{n+1} > 0.
\end{displaymath}
We proved the equation (9.85) in \cite{Un08} and that $\mu_{n, \vv}$ is an $(n+1) \times (n+1)$ real mixed state.
\end{remark}

In \cite[Theorem 5.6]{KL2} it has been shown that
\begin{equation} \label{E:Boost}
\displaystyle B(\vv) = O_{\vv}^{T} \left(
\begin{array}{cc}
  \begin{array}{cc}
    \lambda & 0 \\
    0 & \frac{1}{\lambda} \\
  \end{array} & 0 \\
  0 & I_{n-1} \\
\end{array}
\right) O_{\vv},
\end{equation}
where $O_{\vv}$ is the same orthogonal matrix in Theorem \ref{T:Diag Mobius}. So we obtain the interesting result that M\"obius matrix is the normalized Lorentz boost generated by the vector $2 \otimes \vv$, since
\begin{displaymath}
\tr B(\vv)^2 = \lambda^2 + \frac{1}{\lambda^2} + n-1 = 2 \gamma_{\vv}^2 (1 + \Vert \vv \Vert^2) + n-1 = (n-3) + 4
\gamma_{\vv}^2.
\end{displaymath}

\begin{proposition} \label{P:Mobius matrix}
For each $\vv \in \B_{n}$,
\begin{displaymath}
\displaystyle \mu_{n, \vv} = \frac{1}{\tr B(\vv)^2} B(\vv)^2 = \frac{1}{\tr B(2 \otimes \vv)} B(2 \otimes \vv).
\end{displaymath}
\end{proposition}

\section{Fidelity}

It has been issued how to measure the distance of quantum states represented by density matrices, i.e., positive semidefinite Hermitian matrices with trace $1$. The \emph{fidelity} is one of crucial measurements although it is actually not a metric for quantum states. On the other hand, it is a measure of the closedness of two quantum states, that is, the fidelity is $1$ if and only if two quantum states are identical. Moreover, it does give rise to a useful metric and is able to apply for a variety of research areas in quantum information and computation theory; see \cite{NC} and
\cite[Section 9.2.2]{MZC}.

The fidelity for density matrices $\rho$ and $\sigma$ is defined by
\begin{equation} \label{E:Fidelity}
\displaystyle F(\rho, \sigma) := \tr \sqrt{\rho^{1/2} \sigma
\rho^{1/2}}.
\end{equation}
We review some basic properties of the fidelity.
\begin{lemma} \label{L:Fidelity}
The following are satisfied for any density matrices $\rho$ and $\sigma$.
\begin{itemize}
\item[(i)] $0 \leq F(\rho, \sigma) \leq 1$.
\item[(ii)] $F(\rho, \sigma) = 1$ if and only if $\rho = \sigma$.
\item[(iii)] $F(\rho, \sigma) = F(\sigma, \rho)$.
\item[(iv)] $F(U \rho U^{*}, U \sigma U^{*}) = F(\rho, \sigma)$ for
any unitary $U$.
\end{itemize}
\end{lemma}

The property (iv) of Lemma \ref{L:Fidelity} is called the invariance under unitary congruence transformation, so that the fidelity is basis-independent.

\begin{remark} \label{R:Fidelity}
The fidelity $F$ can be quite difficult to calculate, but it takes a simple form for the $2$-by-$2$ density matrices $\rho$ and $\sigma$: see the equation (8.52) in \cite{Bar},
\begin{equation} \label{E:F1 qubit}
F(\rho, \sigma)^2 = \tr (\rho \sigma) + 2 \sqrt{\det(\rho)
\det(\sigma)}.
\end{equation}
From the equations (9.64) and (9.68) in \cite{Un08} we have alternative expression of the fidelity for the $2$-by-$2$ density matrices $\rho_{\uu}$ and $\rho_{\vv}$, where $\uu, \vv \in \B_{3}$:
\begin{equation} \label{E:F2 qubit}
\displaystyle F(\rho_{\uu}, \rho_{\vv})^2 = \frac{1 + \gamma_{\uu \oplus \vv}}{2 \gamma_{\uu} \gamma_{\vv}} = \frac{1}{2} \left\{ 1 + \uu^{T} \vv + \sqrt{1 - \Vert \uu \Vert^2} \sqrt{1 - \Vert \vv \Vert^2} \right\}.
\end{equation}
One can verify that two equations (\ref{E:F1 qubit}) and (\ref{E:F2 qubit}) are the same.
\end{remark}

We introduce a normalized Lorentz boost to be a density matrix, especially the M\"obius matrix. Let us denote
\begin{displaymath}
\displaystyle \hat{B}(\vv) := \frac{1}{\tr B(\vv)} B(\vv)
\end{displaymath}
for any $\vv \in \B_{n}$. Indeed, $\hat{B}(\vv)$ must be a mixed state, a positive definite density matrix, and by Proposition \ref{P:Mobius matrix} we have $\displaystyle \mu_{n, \vv} = \hat{B}(2 \otimes \vv)$.

We see useful properties to show the main result.
\begin{lemma} \label{L:F scalar}
For any positive semidefinite matrices $A, B$ and $\alpha, \beta > 0$,
\begin{displaymath}
\displaystyle F(\alpha A, \beta B) = \sqrt{\alpha \beta} F(A, B).
\end{displaymath}
\end{lemma}

\begin{lemma} \label{L:TrFid}
For any $\uu, \vv \in (\B_{n}, \oplus)$
\begin{displaymath}
\displaystyle \tr [B(\uu) B(\vv)^2 B(\uu)]^{1/2} = 2 \gamma_{\uu \oplus \vv} + n-1.
\end{displaymath}
\end{lemma}

\begin{proof}
From the equation \eqref{E:Boost} we have
\begin{displaymath}
\displaystyle \tr B(\vv) = \lambda + \frac{1}{\lambda} + n-1 = 2 \gamma_{\vv} + n-1.
\end{displaymath}
So it is proved by the equation \eqref{E:EinBoost}.
\end{proof}

We now see an explicit formula of the fidelity for normalized Lorentz boosts in terms of Lorentz factors.

\begin{theorem} \label{T:Fid boost}
For any $\uu, \vv \in (\B_{n}, \oplus)$
\begin{displaymath}
\displaystyle F(\hat{B}(\uu), \hat{B}(\vv)) = \frac{2 \gamma_{\ww} + n-1}{\sqrt{(2 \gamma_{\uu} + n-1)(2 \gamma_{\vv} + n-1)}},
\end{displaymath}
where
\begin{displaymath}
\displaystyle \ww = \frac{1}{2} \otimes \uu \oplus \frac{1}{2} \otimes \vv.
\end{displaymath}
\end{theorem}

\begin{proof}
Let $\uu' := (1/2) \otimes \uu$ and $\vv' := (1/2) \otimes \vv$. Then
\begin{displaymath}
\begin{split}
\displaystyle F(\hat{B}(\uu), \hat{B}(\vv)) & = \frac{F(B(\uu), B(\vv))}{\sqrt{\tr B(\uu) \tr B(\vv)}} = \frac{\tr [B(\uu') B(\vv')^2 B(\uu')]^{1/2}}{\sqrt{\tr B(\uu) \tr B(\vv)}} \\
& = \frac{2 \gamma_{\uu' \oplus \vv'} + n-1}{\sqrt{(2 \gamma_{\uu} + n-1)(2 \gamma_{\vv} + n-1)}}.
\end{split}
\end{displaymath}
The first equality follows from Lemma \ref{L:F scalar}, the second follows from Remark \ref{R:Boost}, and the last follows from Lemma \ref{L:TrFid}.
\end{proof}

For any $\uu, \vv \in \B_{n}$, in general,
\begin{displaymath}
\displaystyle (1/2) \otimes \uu \oplus (1/2) \otimes \vv \neq (1/2) \otimes (\uu \oplus \vv),
\end{displaymath}
see \cite[Chapter 6]{Un08} for more details. On the other hand, we give an formula of Lorentz factor for $\ww = (1/2) \otimes \uu \oplus (1/2) \otimes \vv$, so that the fidelity for M\"obius matrices can be simply calculated.

\begin{lemma} \label{L:gamma w}
For any $\uu, \vv \in (\B_{n}, \oplus)$,
\begin{displaymath}
\displaystyle \gamma_{\ww} = \frac{1}{\sqrt{1 + 2 \gamma_{\uu}}} \frac{1}{\sqrt{1 + 2 \gamma_{\vv}}} ((1 + \gamma_{\uu})(1 + \gamma_{\vv}) + \gamma_{\uu} \gamma_{\vv} \uu^{T} \vv),
\end{displaymath}
where $\ww = (1/2) \otimes \uu \oplus (1/2) \otimes \vv$.
\end{lemma}

\begin{proof}
Let $\uu' := (1/2) \otimes \uu$ and $\vv' := (1/2) \otimes \vv$. By Remark \ref{R:1/2v}
\begin{displaymath}
\displaystyle \gamma_{\vv'} = \frac{1}{\sqrt{1 - \Vert \vv' \Vert^2}} = \frac{1}{\sqrt{1 - \gamma_{\vv}^2 / (1 +
\gamma_{\vv})^2}} = \frac{1 + \gamma_{\vv}}{\sqrt{1 + 2 \gamma_{\vv}}}.
\end{displaymath}
Applying the gamma identity $\gamma_{\uu \oplus \vv} = \gamma_{\uu} \gamma_{\vv} (1 + \uu^T \vv)$ to $\uu'$ and $\vv'$, it is proved.
\end{proof}

\begin{remark}
Directly from the Einstein velocity addition we have
\begin{displaymath}
\displaystyle 2 \otimes \vv = \vv \oplus \vv = \frac{2 \gamma_{\vv}^2}{2 \gamma_{\vv}^2 - 1} \vv,
\end{displaymath}
so that $\displaystyle \gamma_{2 \otimes \vv} = 2 \gamma_{\vv}^2 - 1$. Hence, the result of Theorem \ref{T:Fid boost} reduces to
\begin{displaymath}
\displaystyle F(\mu_{n, \uu}, \mu_{n, \vv}) = F(\hat{B}(2 \otimes \uu), \hat{B}(2 \otimes \vv)) = \frac{\gamma_{\uu \oplus \vv} + n - 1}{\sqrt{(4 \gamma_{\uu}^{2} + n - 3)(4 \gamma_{\vv}^{2} + n - 3)}}.
\end{displaymath}
Especially, if $n = 3$,
\begin{displaymath}
\displaystyle F(\mu_{3, \uu}, \mu_{3, \vv}) = \frac{1 + \gamma_{\uu \oplus \vv}}{2 \gamma_{\uu} \gamma_{\vv}} = F(\rho_{\uu}, \rho_{\vv})^2.
\end{displaymath}
\end{remark}


\begin{thebibliography}{99}

\bibitem{Bar}
Stephen M. Barnett, Quantum Information, Oxford University Press, 2009.

\bibitem{KL2}
S. Kim and J. Lawson, Unit balls, Lorentz boosts, and hyperbolic geometry, Results Math. \textbf{63} (2013), 1225-1242.

\bibitem{LL1}
J. Lawson and Y. Lim, Symmetric sets with midpoints and algebraically equivalent theories, Results Math. \textbf{46} (1-2), 37-56 (2004).

\bibitem{MZC}
Z. Ma, F.-L. Zhang, and J.-L. Chen, Fidelity induced distance measures for quantum states, Phys. Lett. A \textbf{373} (2009), 3407-3409.

\bibitem{NC}
M. Nielsen and I. Chuang, Quantum Computation and Quantum Information, Cambridge, 2010.

\bibitem{SSS}
L. V. Sabinin, L. L. Sabinina, and L. V. Sbitneva, On the notion of a gyrogroup, Aeq. Math. \textbf{56} (1998), 11-17.

\bibitem{Un08}
A. A. Ungar, Analytic hyperbolic geometry and Albert Einstein's special theory of relativity, World Scientific Press, 2008.

\end{thebibliography}
\end{document}